\documentclass[11pt]{article}
\usepackage[margin=1in]{geometry}                
\usepackage{graphicx}
\usepackage{amssymb}
\usepackage{amsmath}
\usepackage{authblk}
\usepackage{amsthm}
\usepackage{color}
\usepackage{multicol}

\newtheorem{theorem}{Theorem}[section]
\newtheorem{definition}[theorem]{Definition}
\newtheorem{corollary}[theorem]{Corollary}
\newtheorem{lemma}[theorem]{Lemma}
\newtheorem{remark}[theorem]{Remark}

\newcommand{\N}{\mathbb{N}}

\newcommand{\rk}{\mathsf{rk}}
\newcommand{\rka}{\mathsf{ark}}
\newcommand{\Ack}{\mathsf{Ack}}
\newcommand{\p}{\mathcal{P}}

\newcommand{\Oh}{\mathcal{O}}

\begin{document}
\date{}

\title{Enumeration of the adjunctive hierarchy of \\hereditarily finite sets}
\author{Giorgio Audrito}
\affil{Department of Mathematics ``Giuseppe Peano'', University of Torino, Italy\\
\texttt{giorgio.audrito@unito.it}}

\author{Alexandru I. Tomescu}
\affil{Helsinki Institute for Information Technology HIIT\\ Department of Computer Science, University of Helsinki, Finland\\
\texttt{tomescu@cs.helsinki.fi}}

\author{Stephan Wagner}
\affil{Department of Mathematical Sciences, Stellenbosch University, South Africa\\
\texttt{swagner@sun.ac.za}}

\maketitle

\begin{abstract}
Hereditarily finite sets (sets which are finite and have only hereditarily finite sets as members) are basic mathematical and computational objects, and also stand at the basis of some programming languages. We solve an open problem proposed by Kirby in 2008 concerning a recurrence relation for the cardinality $a_n$ of the $n$-th level of the adjunctive hierarchy of hereditarily finite sets; in this hierarchy, new sets are formed by the addition of a new single element drawn from the already existing sets to an already existing set. We also show that our results can be generalized to sets with atoms, or can be refined by rank, cardinality, or by the maximum level from where the new adjoined element is drawn. 

We also show that $a_n$ satisfies the asymptotic formula $a_n = C^{2^n} + O(C^{2^{n-1}})$, for a constant $C \approx 1.3399$, which is a too fast asymptotic growth for practical purposes. We thus propose a very natural variant of the adjunctive hierarchy, whose asymptotic behavior we prove to be $\Theta(2^n)$.

\smallskip
\noindent\textsc{Keywords:} Finite set theory, adjunction, hierarchy, counting problem, combinatorial problem, recurrence relation, asymptotic analysis \\
\noindent \textsc{1998 ACM Subject Classification:} F.4.1 Mathematical Logic; G.2.1 Combinatorics

\end{abstract}

\section{Introduction}
\label{intro}

\begin{table*}[t]
\small
\caption{The first 10 values of $a_n$.\label{table-adj}}
\begin{center}
\begin{tabular}{|l|l|}
\hline
$n$ & $a_n$ \\\hline
0 &  1 \\
1 &  2 \\
2 &  4 \\
3 &  12 \\
4 &  112 \\
5 &  11680 \\
6 &  135717904 \\
7 &  18418552718041816 \\
8 &  339243082977367810522963263986432 \\
9 &  115085869347989258409868700405844845152126435897832396556555233936 \\
\hline
\end{tabular}
\end{center}
\label{table-main}
\end{table*}%

Sets are basic mathematical and computational objects; usually one considers sets of two kinds: \emph{pure sets}, that is, sets which, unless empty, contain only other pure sets as elements; or \emph{sets with atoms}, that is, sets which can also have as elements objects from a given collection of \emph{atoms}, or \emph{urelements}. In this paper we focus mainly on sets which are pure and hereditarily finite, in the sense that they are finite, and all of their members are hereditarily finite sets.

This paper is motivated by an effort to cross-fertilize set theory and computer science, started by Jacob T. Schwartz in the 1970s. This field, called \emph{computable set theory} in \cite{CFO89}, has led, on the one hand, to set-based programming languages such as $\mathsf{SETL}$ \cite{SDDS86}, or the more recent \{log\} \cite{DovierOPR96} and CLP($\mathcal{SET}$) \cite{DovierPPR00}. On the other hand, it has uncovered decidable fragments of set theory \cite{STfC,ParlamentoPolicriti1991}. One emblematic example is the {\bf M}ulti-{\bf L}evel-{\bf S}yllogistic with {\bf S}ingleton fragment and its enaction into the automatic proof-checker {\sf Referee/\AE tnaNova} \cite{REF2,Computational-Logic-and-Set-Theory}. 

These combined efforts have raised the need for \emph{efficient} computer representations of hereditarily finite sets. One such representation can simply be a bijection between hereditarily finite sets and natural numbers; in this case one is also interested in so-called \emph{ranking/unranking} algorithms. By an efficient representation we mean that less `complex' sets have smaller encodings; or, otherwise stated, that sets having `complexity' $n$ be encoded with a number of bits polynomial in $n$. 

Any such representation depends on the measure of `complexity' of the sets. Since the family of hereditarily finite sets, denoted as $V_\omega$, is usually defined by an iterative bottom-up construction starting from the empty set, the `complexity' of a set is the stage of this iterative process at which it is constructed. Consider the usual von Neumann's \emph{cumulative hierarchy} \cite{Jech:2003ly} obtained by repeatedly taking the family of subsets of the sets constructed so far. Formally,
\[V_0 := \emptyset, \ \ V_{n+1} := \mathcal{P}(V_n), \ \ V_\omega := \bigcup_{n \in \omega} V_n,\] 
\noindent where $\mathcal{P}$ denotes the power-set operator. The first levels of this hierarchy are:
\begin{itemize}
\item $V_1 = \left\{\emptyset\right\}$, 
\item $V_2 = \left\{\emptyset, \{\emptyset\}\right\}$, 
\item $V_3 = \big\{ \emptyset, \{\emptyset\}, \left\{\{\emptyset\}\right\}, \{\emptyset, \{\emptyset\}\} \big\}$.
\end{itemize}

In the case of the cumulative hierarchy, the `complexity' of a set $x$ is that level $n$ such that $x \in V_{n} \setminus V_{n-1}$; number $n-1$ is also called the \emph{rank} of the set $x$ \cite{Jech:2003ly}. An encoding of hereditarily finite sets w.r.t.~the cumulative hierarchy must place the sets in $V_{n} \setminus V_{n-1}$ after all the sets in $V_{n-1}$, for all $n \geq 1$. One such encoding is the classical Ackermann's encoding \cite{Ack37}, recursively defined as $\Ack(\emptyset) = 0$ and $\Ack(x) = \sum_{y \in x}2^{\Ack(y)}$. However, since $|V_{n}| = 2^{|V_{n-1}|}$ holds for all $n  \geq 1$, any encoding of hereditarily finite sets w.r.t.~the cumulative hierarchy is not feasible in practice---as also noted in \cite{Chiaruttini-Omodeo-08,Kirby}---, since one needs a super-exponential number of bits to represent a set in $V_{n} \setminus V_{n-1}$.

An alternative approach is a combinatorial one. In order to avoid further complications, let us  restrict our exposition to \emph{transitive} sets (a set $x$ is \emph{transitive} if all elements of $x$ are also subsets of $x$), and denote by $T_n$ the family of transitive hereditarily finite sets with at most $n$ elements. In the case of the hierarchy $T_n$, the `complexity' of a transitive hereditarily finite set $x$ is just its cardinality. For this hierarchy, an efficient representation exists, since $|T_n \setminus T_{n-1}|$ is $O(2^{n^2})$, and any set in $T_n \setminus T_{n-1}$ can be ranked/unranked using $O(n^5)$ bit operations \cite{DBLP:journals/ipl/RizziT13}. This result uses a recurrence relation for $|T_n \setminus T_{n-1}|$; the enumeration problem was initially solved in \cite{P62} (see also \cite[p. 123]{Comtet74}), and also by a different method in \cite{DBLP:journals/ipl/PolicritiT11}. Moreover, the asymptotic behavior of this recurrence was recently determined in \cite{DBLP:conf/analco/Wagner12,Wagner:2012fk}.

Even though satisfactory from a combinatorial point of view, this approach lacks set theoretic motivation, as the `complexity' of a set does not reflect a natural set theoretic operation leading to its production (as the cumulative hierarchy does). A possible solution was proposed by Kirby in \cite{Kirby}. Kirby used the well-known fact that hereditarily finite sets can also be obtained from the empty set by repeated use of the \emph{adjunction} (or \emph{adduction}) operator \cite{givant}:
\[\langle x,y \rangle \mapsto x \cup \{y\}.\]
In view of this fact, Kirby proposed the following natural hierarchy of hereditarily finite sets, which he called the \emph{adjunctive hierarchy}:
\[A_0 := \{\emptyset\}, \ \ A_{n+1} := \{\emptyset\} \cup \big\{ x \cup \{y\} \;|\; x,y \in A_n \big\},\]
so that $A_n \subseteq A_{n+1}$ for all $n \in \omega$, and $\bigcup_{n \in \omega}A_n = V_\omega$.\footnote{To be precise, in \cite{Kirby}, the $(n+1)$th level of the hierarchy was defined as $A'_{n+1} := A'_n \cup \left\{ x \cup \{y\} \;|\; x,y \in A'_n \right\}$, where $A'_0 = \{\emptyset\}$; this is equivalent to our definition.} Here, the `complexity' of a set $x$ is analogously taken to be that level $n$ such that either $x \in A_0$ and $n = 0$, or $x \in A_n \setminus A_{n-1}$. Observe that
\begin{itemize}
\item $A_0 = V_1$, $A_1 = V_2$, $A_2 = V_3$, but
\item $|A_3| = 12 < 2^{2^2} = |V_4|$, $|A_4| = 112 < 2^{2^{2^{2}}} = |V_5|$.
\end{itemize}

Thus, the adjunctive hierarchy grows more slowly than the cumulative hierarchy. However, in order to ascertain the feasibility of an encoding of hereditarily finite sets w.r.t.~the adjunctive hierarchy, it is crucial to exactly determine $a_n := |A_n|$, for every $n \in \omega$; this problem was left open in \cite{Kirby}, where the values of $a_n$ for $n \leq 6$ were obtained. 


To begin with, in Sec.~\ref{sec-main} we give a compact recurrence relation for $a_n$, for any $n \in \omega$. This recurrence relation can be implemented by a simple dynamic programming algorithm, which runs in polynomial time in the arithmetic model, thus allowing a fast computation of the values $a_n$ (the values for $n < 10$ are given in Table~\ref{table-main}). Moreover, our method allows us to impose restrictions on the sets making up a level of the adjunctive hierarchy, for example on rank, or on cardinality; this is presented in Sec.~\ref{sec:rank} and~\ref{sec:cardinality}. In Sec.~\ref{sec:urelements} we argue that all of our recurrence relations can be generalized to an analogous adjunctive hierarchy of hereditarily finite sets with atoms, in which $A_0$ also includes an arbitrary finite set of atoms, and $A_0$ is added to every layer of the adjunctive hierarchy.

As we will show in Sec.~\ref{sec:asymptotics}, the adjunctive hierarchy is not able to meet the goal of a hierarchy allowing an efficient encoding of hereditarily finite sets, since $a_n$ is seen to have a number of bits exponential in $n$. More precisely, we prove that $a_n \sim C^{2^n}$ for a constant $C \approx 1.3399$. Nevertheless, in Sec.~\ref{sec:bounded}, we show that the growth of the adjunctive hierarchy can be controlled by any unbounded sublinear function $f$ limiting the maximum level of the hierarchy from where the new adjoined element can be drawn. Our recurrence relation easily generalizes to this context, and in the Appendix we give numerical values for different choices of $f$. This restriction turns out to be a natural lever for slowing down the asymptotic growth of the hierarchy $A_n$.

More importantly, we identify a natural hierarchy $\bar{A}_n$, which we call the \emph{minimally bounded adjunctive hierarchy}, whose asymptotic behavior we prove to be $\Theta(2^n)$. In this hierarchy, in order to make up the $(n+1)$th level, the adjunction of a set $y \in \bar{A}_{m+1}$ to a set $x \in \bar{A}_n$ is allowed only if all the sets in $\p(\bar{A}_m)$ are already present in $\bar{A}_n$; that is, when the adjunction of any set in $\bar{A}_{m}$ to sets in $\bar{A}_n$ does not produce any new set. To the best of our knowledge, this hierarchy is the first to have a slow growth and a natural definition, as it can be seen as a  combination of the adjunctive and cumulative hierarchies. 


\section{The adjunctive hierarchy}
\label{sec-main}

\subsection{Adjunctive rank and basic properties}

We start by presenting some remarks and basic properties that will be used to prove the main result.

\begin{remark}[\!\!\cite{Kirby}]
\label{remark-kirby}
If $x \in A_n$, then $|x| \leq n$, $x \subseteq A_{n-1}$ and for any $y \subseteq x$, $y \in A_n$ holds.
\end{remark}

\begin{remark}
\label{remark-basic}
If $x \in A_n \setminus A_{n-1}$, then $x \subseteq A_{n-1}$ and there exists a $y \in x$ such that $x \setminus \{y\} \in A_{n-1}$.
\end{remark}

The \emph{rank} of a hereditarily finite set $x$ is a well-known measure of its complexity \cite{Jech:2003ly}, and can be defined as 
\[\rk(x) := \min\{n :~ x \in V_n\} - 1.\]

Equivalently stated, the rank of $x$ is either 0, if $x = \emptyset$, or otherwise it can be recursively expressed as
\[\rk(x) := \max\{\rk(y) :~ y \in x\} + 1.\]

In analogy with the usual cumulative hierarchy, we introduce the following definition.

\begin{definition}
For every hereditarily finite set $x$, the \emph{adjunctive rank} of $x$ is defined as
\[\rka(x) := \min\{n :~ x \in A_n\}.\]
Equivalently, $\rka(x)$ is either 0, if $x = \emptyset$, or it is the number $n$ such that $x \in A_n \setminus A_{n-1}$.
\end{definition}

As in the cumulative case, we are interested in finding an equivalent recursive definition of $\rka(x)$ using the set $\{\rka(y): ~ y \in x \}$. This is possible, even if slightly trickier than in the classical case.

\begin{remark}
\label{remark-rka0}
For all hereditarily finite sets $x,y$, the following holds:
\begin{align*}
\max\{\rka(x),\rka(y)\} \leq \rka(x \cup \{y\}) \leq \max\{\rka(x),\rka(y)\} + 1.
\end{align*}
\end{remark}

\begin{remark}
\label{remark-rka1}
If $\rka(x) \leq \rka(y)$, then $\rka(x \cup \{y\}) = \max\{\rka(x),\rka(y)\} + 1 = \rka(y) + 1$.
\end{remark}

Unfortunately, it is not true that when $\rka(x) \geq \rka(y)$ the adjunctive rank of $x \cup \{y\}$ is $\rka(x) + 1$ (e.g.~when $x = \{\{\emptyset\}\}$, $y=\emptyset$). However, for any hereditarily finite set $z$, this can be avoided for a specific choice of sets $x$, $y$ such that $x \cup \{y\} = z$, and this will allow us to state the following recursive definition of $\rka(x)$ from $\{\rka(y): ~ y \in x\}$.

\begin{lemma}
\label{lemma-rka}
If $y = \{x_1, \ldots ,x_n\}$ such that $\rka(x_i) \leq \rka(x_{i+1})$ for all $i \in \{1,\dots,n-1\}$, then $\rka(y) = \max\{\rka(x_j) + n - j:~ 1 \leq j \leq n\} + 1$.
\end{lemma}

\begin{proof}
We prove the statement by induction on $n$. If $n=1$, the thesis follows from Remark~\ref{remark-rka1} for $\emptyset \cup \{x_1\}$.

Now suppose that $n>1$. From the inductive hypothesis, we have 
\begin{align*}
\rka(y \setminus \{x_n\}) &= \max\{\rka(x_j) + n -1 - j:~ 1 \leq j \leq n-1\} + 1\\
&= \max\{\rka(x_j) + n - j:~ 1 \leq j \leq n-1\}.
\end{align*}

If $\rka(y \setminus \{x_n\}) \leq \rka(x_n)$, from Remark~\ref{remark-rka1} we have 
\begin{align*}
\rka(y) &= \max\left\{\rka(y \setminus \{x_n\}),\rka(x_n) \right\} + 1 \\
&= \max\big\{\max\{\rka(x_j) + n - j :~ 1 \leq j \leq n-1\}, \\
& \phantom{\hspace{.15cm}= \max\big\{}\rka(x_n) + n - n \big\} + 1\\
&= \max\left\{\rka(x_j) + n - j:~ 1 \leq j \leq n\right\} + 1,
\end{align*}
which is the thesis.

Otherwise, $\rka(y \setminus \{x_n\}) > \rka(x_n)$. From Remark~\ref{remark-basic}, we know that there exists an $i$ such that $\rka(y \setminus \{x_i\}) < \rka(y)$. Applying Remark~\ref{remark-rka0} together with $\max\{\rka(y \setminus \{x_i\}), \rka(x_i)\} < \rka(y)$, it follows that
\[
\begin{array}{rcl}
\rka(y) &=& \max\{\rka(y \setminus \{x_i\}), \rka(x_i)\} + 1.
\end{array}
\]
By the inductive hypothesis on $y \setminus \{x_i\}$, we have that 

\begin{align*}
\rka(y \setminus \{x_i\}) = \max\hspace{-.07cm}\big\{&\{\rka(x_j) + n - j :~ 1 \leq j \leq i-1\} \; \cup \\
&\{\rka(x_{j+1}) + n - j:~ i \leq j \leq n-1\}\big\}.
\end{align*}

Since for all $j \in \{i,\dots,n-1\}$, $\rka(x_{j+1}) + n - j \geq \rka(x_j) + n - j$ holds from the assumption that the $x_i$'s form a non-decreasing sequence, this implies that $\rka(y \setminus \{x_i\}) \geq \rka(y \setminus \{x_n\})$. Thus, by the hypothesis on $\rka(x_n)$, it holds that 
\[
\rka(y \setminus \{x_i\}) \geq \rka(y \setminus \{x_n\}) > \rka(x_n) \geq \rka(x_i).
\]

Combining the last inequalities with our assumption on $i$,
\[
\begin{array}{rcl}
\rka(y) &=& \max\{\rka(y \setminus \{x_i\}), \rka(x_i)\} + 1 \\
&=& \rka(y \setminus \{x_i\}) + 1 \\
&\geq& \rka(y \setminus \{x_n\}) + 1 \\
&=&\max\{\rka(y \setminus \{x_n\}),\rka(x_n)\} + 1\\
&=&\max\{\rka(x_j) + n - j :~ 1 \leq j \leq n\} + 1
\end{array}
\]
which, by Remark~\ref{remark-rka0}, proves the thesis.
\end{proof}

In the next sections, Lemma~\ref{lemma-rka} will be used mainly in the form of the following corollary.

\begin{corollary}
\label{corollary-rka}
If $x \subseteq A_{\rka(y)}$, then $\rka(x \cup \{y\}) = \max\{\rka(x),\rka(y)\} + 1$.
\end{corollary}

\subsection{The recurrence relation for $a_n$}

To develop a recurrence relation for the sequence $a_n$ we first need to define a finer sequence, as in the following definition.

\begin{definition}
For every $n > m \geq 0$, let $B_{n,m}$ be the set $\left\lbrace x \in A_n \setminus A_{n-1} : ~ x \subseteq A_m \right\rbrace$, and let $b_{n,m} := |B_{n,m}|$.
\end{definition}

Notice that $B_{n,n-1} = A_{n} \setminus A_{n-1}$ and $B_{n,0} = \emptyset$ for all $n > 1$; we shall extend for convenience the definitions of $B_{n,m}$ and, accordingly, of $b_{n,m}$, to all integers $n,m$, by assuming $A_n = \emptyset$ if $n<0$. Therefore, we can write
\[a_{n} = \sum_{k = 0}^{n}b_{k,k-1}.\]
For the sequence $b_{n,m}$ a compact recurrence relation can be provided.

\begin{theorem}
\label{thm-main}
For all natural numbers $n > m \geq 0$, the following recurrence relation holds
\begin{align}
b_{n,m} &= b_{n,m-1} + \sum_{k=1}^{n-m-1} \left( b_{n-k,m-1} \binom{b_{m,m-1}}{k} \right) + \binom{b_{m,m-1}}{n-m} \sum_{k=0}^m b_{k,k-1},\label{eq:rec}
\end{align}
where $b_{0,-1} = 1$, and $b_{n,-1} = 0$, for all $n \geq 1$ (we assume that $\binom{a}{b} = 0$ if $a < b$).
\end{theorem}

\begin{proof}
The recurrence relation above follows from partitioning the set $B_{n,m}$ into sets $B_{n,m,k}$, for $k \in \{0,\dots,n-m\}$:
\[
B_{n,m,k} := \left\lbrace x \in B_{n,m} : ~ |x \cap (A_m \setminus A_{m-1})| = k \right\rbrace.
\]

This is indeed a partition of $B_{n,m}$, since for any $x \in B_{n,m}$, $|x \cap (A_m \setminus A_{m-1})| \leq n - m$ holds. Otherwise, denote $|x \cap (A_m \setminus A_{m-1})|$ by $l$, and suppose that $l \geq n-m+1$. Therefore, $x \subseteq A_m$ is obtained by adjoining $l$ elements in $A_m \setminus A_{m-1}$ to a set in some $A_d$; by Corollary~\ref{corollary-rka} (applied $l$ times), we have that $d \leq n-l \leq m-1$. However, by the first adjunction of such an element in $A_m \setminus A_{m-1}$, we obtain, by Remark~\ref{remark-rka1}, a set in $A_{m+1} \setminus A_m$. By the adjunction of the remaining $l - 1 \geq n-m$ elements to it, we obtain, by Corollary~\ref{corollary-rka}, that the set $x$ is in $A_{l+m} \setminus A_{l+m-1}$. Since $l + m \geq n + 1$, this contradicts the assumption that $x \in B_{n,m}$.

By definition, $B_{n,m,0} = B_{n,m-1}$, which gives the first term in the recurrence relation.

For $k \in \{1,\dots,n-m-1\}$, every set $x \in B_{n,m,k}$ is obtained by adjoining $k$ sets in $A_m \setminus A_{m-1}$, that can be chosen in $\binom{b_{m,m-1}}{k}$ ways, to a set $y \in A_{n-k}$ such that $y \subseteq A_{m-1}$. Observe that $y \notin A_{n-k-1}$, since adjoining $k$ elements in $A_m$, where $m \leq n-k-1$, to a set in $A_{n-k-1}$ would result in a set in $A_{n-1}$, by the definition of the hierarchy. Therefore, such sets $y$ are precisely those in $B_{n-k,m-1}$, hence $|B_{n,m,k}| = b_{n-k,m-1} \binom{b_{m,m-1}}{k}$, giving the next term in the recurrence relation.

Similarly, if $k = n-m$, every set $x \in B_{n,m,n-m}$ is obtained adjoining $n-m$ sets in $A_m \setminus A_{m-1}$, that can be chosen in $\binom{b_{m,m-1}}{n-m}$ ways, to a set $y \in A_m$. In this case, however, any set in $A_m$ can be used as $y$, since adjoining $n-m$ sets in $A_m \setminus A_{m-1}$ to a set in $A_m$ will always produce a set in $A_n \setminus A_{n-1}$, by Corollary~\ref{corollary-rka}. This gives the last term in the recurrence relation.
\end{proof}

This result can be implemented in an algorithm, to obtain the numeric values shown in Table~\ref{table-adj}.

%

\section{Asymptotic behavior}
\label{sec:asymptotics}

In this section, we are interested in the asymptotic behavior of the sum
$$a_n = \sum_{k=0}^n b_{k,k-1}.$$
To this end, we first study the asymptotics of $c_n := b_{n,n-1}$. First we note that
\begin{equation}\label{eq:increasing}
b_{n,-1} \leq b_{n,0} \leq \cdots \leq b_{n,n-1}
\end{equation}
by definition (since $B_{n,m-1} \subseteq B_{n,m}$), which can also be seen from the fact that all the terms in~\eqref{eq:rec} are nonnegative. Moreover, for $m = n-1$, one of the terms on the right hand side of~\eqref{eq:rec} is $b_{n-1,n-2}^2$ (corresponding to $k=m=n-1$ in the second sum), so that
$$c_n = b_{n,n-1} \geq b_{n-1,n-2}^2 = c_{n-1}^2.$$
It turns out that this is in fact the dominant term (combinatorially speaking, this means that ``most'' of the elements of $A_n \setminus A_{n-1}$ are obtained by adjoining an element of $A_{n-1} \setminus A_{n-2}$ to another set in $A_{n-1} \setminus A_{n-2}$), as the following lemma shows:

\begin{lemma}\label{lem:main_lemma}
The sequence $c_n = b_{n,n-1}$ satisfies
$$c_n = c_{n-1}^2 \left( 1 + \Oh(1/c_{n-2}) \right).$$
Specifically,
$$c_{n-1}^2 \leq c_n \leq c_{n-1}^2 \left( 1 + \frac{4}{c_{n-2}} \right)$$
for $n \geq 2$.
\end{lemma}

\begin{proof}
The inequality $c_{n-1}^2 \leq c_n$ has already been mentioned. It also follows from this inequality that $c_{n-k}^{2^k} \leq c_n$ for $0 \leq k \leq n$, and since $c_2 = 2$, we have $c_n \geq 2^{2^{n-2}}$ for all $n \geq 2$ by a simple induction, which means that $c_n$ grows quite rapidly. Finally, we have $c_n \geq c_{n-1}^2 \geq 2c_{n-1}$ for $n \geq 3$ (and also $c_2 = 2 \geq 2 = 2c_1$), hence $2^k c_{n-k} \leq c_n$ for $n \geq 2$ and $k \leq n-1$.

It remains to prove the second inequality. The special case $m=n-1$ in~\eqref{eq:rec} yields
$$c_n = b_{n,n-1} = b_{n-1,n-2} \sum_{k=0}^{n-1} b_{k,k-1} + b_{n,n-2}$$
and thus also (replacing $n$ by $n-1$)
\begin{align}
c_{n-1} &= b_{n-1,n-2} \nonumber\\
&= b_{n-2,n-3} \sum_{k=0}^{n-2} b_{k,k-1} + b_{n-1,n-3} \nonumber\\
& \geq b_{n-2,n-3} \sum_{k=0}^{n-2} b_{k,k-1},\label{eq:sum_ineq}
\end{align}
from which we obtain
\begin{align*}
c_n &\leq b_{n-1,n-2} \left( b_{n-1,n-2} + \frac{b_{n-1,n-2}}{b_{n-2,n-3}} \right) + b_{n,n-2} \\
& = c_{n-1}^2 \left( 1 + \frac{1}{c_{n-2}} \right) + b_{n,n-2}.
\end{align*}
Moreover,~\eqref{eq:sum_ineq} implies
\begin{equation}\label{eq:sumest}
\sum_{k=0}^{n-2} c_k = \sum_{k=0}^{n-2} b_{k,k-1} \leq \frac{c_{n-1}}{c_{n-2}},
\end{equation}
which we will use later. To complete the proof, we have to show that
$$b_{n,n-2} \leq \frac{3c_{n-1}^2}{c_{n-2}}.$$
To this end, we iterate our recursion~\eqref{eq:rec} to obtain
\begin{align*}
b_{n,n-2} &=  b_{n,n-3} + \sum_{k=1}^{1} \left( b_{n-k,n-3} \binom{b_{n-2,n-3}}{k} \right) + \binom{b_{n-2,n-3}}{2} \sum_{k=0}^{n-2} b_{k,k-1} \\
&= b_{n,n-4} + \cdots \\
&= \sum_{r=0}^{n-2} \sum_{k=1}^{n-r-1} \left( b_{n-k,r-1} \binom{b_{r,r-1}}{k} \right) + \sum_{r=0}^{n-2} \binom{b_{r,r-1}}{n-r} \sum_{k=0}^r b_{k,k-1}.
\end{align*}
We split this into three parts: $k=1$ in the first sum gives us
$$S_1 = \sum_{r=0}^{n-2} b_{n-1,r-1} b_{r,r-1} \leq b_{n-1,n-2} \sum_{r=0}^{n-2} b_{r,r-1} \leq \frac{c_{n-1}^2}{c_{n-2}}$$
by~\eqref{eq:increasing} and~\eqref{eq:sumest}. The other terms of the first sum taken together can be estimated as follows (note that $\binom{b_{r,r-1}}{k} = 0$ for $r=0$ or $r=1$ now, since $b_{0,-1} = b_{1,0} =1$):
\begin{align*}
S_2 &= \sum_{r=0}^{n-2} \sum_{k=2}^{n-r-1} \left( b_{n-k,r-1} \binom{b_{r,r-1}}{k} \right) \\
&= \sum_{k=2}^{n-3} \sum_{r=2}^{n-k-1} \left( b_{n-k,r-1} \binom{b_{r,r-1}}{k} \right) \\
&\leq \sum_{k=2}^{n-3} \sum_{r=2}^{n-k-1} c_{n-k} \frac{c_r^k}{k!},
\end{align*}
since $b_{n-k,r-1} \leq b_{n-k,n-k-1} = c_{n-k}$ by~\eqref{eq:increasing} and $b_{r,r-1} = c_r$.
Making use of the inequalities $c_{n-k}^{2^k} \leq c_n$ and $2^k c_{n-k} \leq c_n$ mentioned at the beginning of the proof, it follows that
\begin{align*}
S_2 &\leq \sum_{k=2}^{n-3} \frac{c_{n-k}}{k!} \sum_{r=2}^{n-k-1} \left( c_{n-k-1} 2^{r+k-n+1} \right)^k \\
&\leq \sum_{k=2}^{n-3} \frac{c_{n-k}}{k!} c_{n-k-1}^k \sum_{\ell=0}^{\infty} 2^{-\ell k} \\
&\leq \sum_{k=2}^{n-3} \frac{c_{n-2}}{k!} c_{n-k-1}^{2^{k-1}} \cdot \frac{1}{1-2^{-k}} \\
&\leq c_{n-2} c_{n-2} \sum_{k=2}^{\infty} \frac{1}{k!(1-2^{-k})} \\
&< c_{n-2}^2 \leq \frac{c_{n-1}^2}{c_{n-2}}.
\end{align*}
Finally, we have
\begin{align*}
S_3 &= \sum_{r=0}^{n-2}\binom{b_{r,r-1}}{n-r} \sum_{k=0}^r b_{k,k-1} \\
& = \sum_{k=0}^{n-2} \sum_{r=k}^{n-2} \binom{c_r}{n-r} c_k \leq  \sum_{k=0}^{n-2} \sum_{r=0}^{n-2} \frac{c_r^{n-r}}{(n-r)!} c_k \\
&\leq \sum_{k=0}^{n-2} c_k \sum_{r=0}^{n-2} \frac{c_r^{2^{n-r-1}}}{(n-r)!}
\leq \sum_{k=0}^{n-2} c_k \sum_{r=0}^{n-2} \frac{c_{n-1}}{(n-r)!} \\
&\leq c_{n-1} \sum_{k=0}^{n-2} c_k \sum_{\ell=2}^{\infty} \frac{1}{\ell!}
< c_{n-1} \sum_{k=0}^{n-2} c_k \leq \frac{c_{n-1}^2}{c_{n-2}},
\end{align*}
again by~\eqref{eq:sum_ineq} and the inequalities $c_{n-k}^{2^k} \leq c_n$ and $2^k c_{n-k} \leq c_n$. The desired inequality follows by adding the three parts $S_1$, $S_2$ and $S_3$.
\end{proof}

Now we are ready to prove the main asymptotic formula:

\begin{theorem}\label{thm:asymp}
We have
$$c_n = C^{2^n} + \Oh \left(C^{2^{n-1}}\right)$$
for a constant $C \approx 1.339899757746$.
\end{theorem}

\begin{proof}
Set $u_n = \log c_n$. By Lemma~\ref{lem:main_lemma}, we have
$$u_n = 2 u_{n-1} + r_n,$$
with $0 \leq r_n \leq \log \left( 1 + 4/c_{n-2} \right)$ for $n \geq 2$. Iterating yields
\begin{align*}
u_n &= 2^{n-1} u_1 + \sum_{k=2}^n 2^{n-k}r_k = 2^n \sum_{k=2}^n 2^{-k} r_k \\
&= 2^n \sum_{k=2}^{\infty} 2^{-k} r_k - \sum_{k=n+1}^{\infty} 2^{n-k} r_k \\
& = 2^n \sum_{k=2}^{\infty} 2^{-k} r_k - \sum_{\ell=1}^{\infty} 2^{-\ell} r_{n+\ell}.
\end{align*}
Set $$C = \exp \left( \sum_{k=2}^{\infty} 2^{-k} r_k \right)\approx 1.339899757746.$$
Since $0 \leq r_{n+\ell} \leq \log (1 + 4/c_{n+\ell-2} ) \leq \log (1+4/c_{n-1})$ for $\ell \geq 1$, we have
$$0 \leq \sum_{\ell=1}^{\infty} 2^{-\ell} r_{n+\ell} \leq \sum_{\ell=1}^{\infty} 2^{-\ell} \log (1+4/c_{n-1}) = \log (1+4/c_{n-1}),$$
thus
$$2^n \log C \geq u_n \geq 2^n \log C - \log (1+4/c_{n-1})$$
and consequently
$$C^{2^n} \geq c_n = e^{u_n} \geq C^{2^n} \left(1 + \frac{4}{c_{n-1}} \right)^{-1}.$$
So $c_n \sim C^{2^n}$, and more precisely
$$c_n = C^{2^n} \left( 1 + \Oh \left( C^{-2^{n-1}} \right) \right) = C^{2^n} + \Oh \left( C^{2^{n-1}} \right),$$
completing our proof.
\end{proof}
The following corollary is now immediate:

\begin{corollary}
The sequence $a_n = \sum_{k=0}^n b_{k,k-1} = \sum_{k=0}^n c_k$ is asymptotically given  by
$$a_n = C^{2^n} + \Oh \left( C^{2^{n-1}} \right).$$
\end{corollary}
\begin{proof}
We have
\begin{align*}
a_n &= c_n + \sum_{k = 0}^{n-1} c_k = c_n + \Oh \left( \sum_{k=0}^{n-1} C^{2^k} \right) \\
&= c_n + \Oh \left( C^{2^{n-1}} \sum_{k=0}^{n-1} C^{2^k-2^{n-1}} \right) \\
&= c_n + \Oh \left( C^{2^{n-1}} \sum_{\ell=0}^{\infty} C^{-\ell} \right) = C^{2^n} + \Oh \left( C^{2^{n-1}} \right).
\end{align*}
\end{proof}

\begin{remark}
The series representation
$$C = \exp \left( \sum_{k=2}^{\infty} 2^{-k} r_k \right)$$
for the constant $C$ converges quite rapidly, so it can be computed with high accuracy.
\end{remark}

\section{Refinements and generalizations}

\begin{table*}[t]
\small
\caption{Small values of $r^t_n$.\label{table-rank}}
\begin{center}
\begin{tabular}{|l|llllllll|}
\hline
$n$ & $r^0_n$&$r^1_n$& $r^2_n$ & $r^3_n$ & $r^4_n$ & $r^5_n$ & $r^6_n$ & $r^7_n$\\
\hline
0  &	1 & & & & & & & \\
1  &	1 &	1 & & & & & & \\
2  &	1 &	1 &	2 & & & & & \\
3  &	1 &	1 &	2 &	8 & & & & \\
4  &	1 &	1 &	2 &	12 &	96 & & & \\
5  &	1 &	1 &	2 &	12 &	912 &	10752 & & \\
6  &	1 &	1 &	2 &	12 &	3840 &	10130688 &	125583360 & \\
7  &	1 &	1 &	2 &	12 &	10696 &	34070972672 &	1374608250580992 &	17043910396477440 \\\hline
\end{tabular}
\end{center}
\end{table*}%

\subsection{Refining by rank}
\label{sec:rank}

The value $b_{n,n-1} = |A_n \setminus A_{n-1}|$ computed in the previous section is the number of hereditarily finite sets in $A_n$ having adjunctive rank $n$. Therefore, it is interesting to analogously determine how many elements in $A_n$ have a certain rank in the classical sense.

Observe first that for any $x \in A_n$, $\rk(x) \leq n$ holds. With this purpose we introduce the following definition.


\begin{definition}
For every $n > m \geq 0$, and every $0 \leq t \leq m+1$, let $R_{n,m}^{t} := \{x \in A_{n} \setminus A_{n-1} :~ x \subseteq A_{m} \wedge \rk(x) \leq t\}$, and let $r_{n,m}^{t} := |R_{n,m}^{t}|$.
\end{definition}

Let $r_{n}^t := |\{x \in A_{n} :~ \rk(x) = t\}|$, and observe that $r_{n}^t = \sum_{m=0}^n ( r^t_{m,m-1} - r^{t-1}_{m,m-1}) $. The values $r_{n,m}^t$ satisfy a relation completely analogous to the one in Theorem~\ref{thm-main}; numerical values of $r_n^t$ for small values of $n$ and $t$ are shown in Table~\ref{table-rank}.

\begin{theorem}
\label{thm-rank}
For all natural numbers $n > m \geq 0$, and $0 \leq t \leq m+1$, the following recurrence relation holds
\begin{align*}
r^t_{n,m} &= r^t_{n,m-1} + \sum_{k=1}^{n-m-1} \left( r^t_{n-k,m-1} \binom{r^{t-1}_{m,m-1}}{k} \right) + \binom{r^{t-1}_{m,m-1}}{n-m} \sum_{k=0}^m r^t_{k,k-1},
\end{align*}
where $r^t_{0,-1} = 1$, and $r^t_{n,-1} = 0$, for all $n \geq 1$ and $t \geq 0$.
\end{theorem}

\begin{proof}
The proof follows step by step the one of Theorem~\ref{thm-main}. Since $R^t_{n,m} \subseteq B_{n,m}$, we can partition the set $R_{n,m}^t$ into sets $R_{n,m,k}^t$, for $k \in \{0,\dots,n-m\}$:
\[
R_{n,m,k}^t := \left\lbrace x \in R_{n,m}^t : ~ |x \cap (A_m \setminus A_{m-1})| = k \right\rbrace.
\]
Again, the first term of the recurrence comes from the fact that $R_{n,m,0}^t = R_{n,m-1}^t$.

For $k \in \{1,\dots,n-m-1\}$, every set $x \in R^t_{n,m,k}$ is obtained by adjoining $k$ sets in $A_m \setminus A_{m-1}$ with rank at most $t-1$ (hence in $R^{t-1}_{m,m-1}$), to a set $y \in A_{n-k}$ such that $y \subseteq A_{m-1}$ with rank at most $t$. As before, $y \notin A_{n-k-1}$ hence such sets $y$ are precisely those in $R^t_{n-k,m-1}$. This gives the next term in the recurrence relation.

Similarly, if $k = n-m$ every set $x \in R^t_{n,m,n-m}$ is obtained by adjoining $n-m$ sets in $A_m \setminus A_{m-1}$ with rank at most $t-1$ (hence in $R^{t-1}_{m,m-1}$), to any set $y \in A_m$ with rank at most $t$ (by Corollary~\ref{corollary-rka}). This gives the last term in the recurrence relation.
\end{proof}

\subsection{Refining by cardinality}
\label{sec:cardinality}

\begin{table*}[t]
\caption{Small values of $d^t_n$.\label{table-cardinality}}
\begin{center}
\footnotesize
\setlength{\tabcolsep}{.16cm}
\begin{tabular}{|l|lllllll|}
\hline
$n$ & $d^0_n$ & $d^1_n$ & $d^2_n$ & $d^3_n$ & $d^4_n$ & $d^5_n$ & $d^6_n$ \\
\hline
0 &	1 & & & & & & \\
1 &	1 &	1 & & & & & \\
2 &	1 &	2 &	1 & & & & \\
3 &	1 &	4 &	5 &	2 & & & \\
4 &	1 &	12 &	38 &	44 &	17 & & \\
5 &	1 &	112 &	1266 &	3964 &	4573 &	1764 & \\
6 &	1 &	11680 &	1301832 &	14711308 &	46060477 &	53135964 &	20496642\\
\hline
\end{tabular}
\end{center}
\end{table*}%

In this section we extend Theorem~\ref{thm-main} by imposing a restriction on the cardinality of the sets in some $A_n$. Observe that for any $x \in A_n$, $|x| \leq n$ holds.

%

\begin{definition}
For every $n > m \geq 0$, and every $0 \leq t \leq n$, let $D_{n,m}^{t} := \{x \in A_{n} \setminus A_{n-1} :~ x \subseteq A_{m} \wedge |x| \leq t\}$, and let $d_{n,m}^{t} := |D_{n,m}^{t}|$.
\end{definition}

As in the previous section, we observe that $d_{n}^t := |\{x \in A_{n} :~ |x| = t\}|$ equals $\sum_{m=0}^n (d^t_{m,m-1} - d^{t-1}_{m,m-1})$, and the values $d_{n,m}^t$ satisfy a recurrence relation similar to the one in Theorem~\ref{thm-main}. Numerical values of $d_n^t$ are shown in Table~\ref{table-cardinality}.

\begin{theorem}
\label{thm-cardinality}
For all natural numbers $n > m \geq 0$, and $0 \leq t \leq n$, the following recurrence relation holds
\begin{align*}
d^t_{n,m} &= d^t_{n,m-1} + \sum_{k=1}^{n-m-1} \left( d^{t-k}_{n-k,m-1} \binom{b_{m,m-1}}{k} \right) + \binom{b_{m,m-1}}{n-m} \sum_{k=0}^m d^{t-n+m}_{k,k-1},
\end{align*}
where $d^t_{0,-1} = 1$, and $d^t_{n,-1} = 0$, for all $n \geq 1$ and $t \geq 0$.
\end{theorem}

\begin{proof}
The proof follows step by step the one of Theorem~\ref{thm-main}. Since $D^t_{n,m} \subseteq B_{n,m}$, we can partition the set $D_{n,m}^t$ into sets $D_{n,m,k}^t$, for $k \in \{0,\dots,n-m\}$:
\[
D_{n,m,k}^t := \left\lbrace x \in D_{n,m}^t : ~ |x \cap (A_m \setminus A_{m-1})| = k \right\rbrace.
\]
Again, the first term of the recurrence comes from the fact that $D_{n,m,0}^t = D_{n,m-1}^t$.

For $k \in \{1,\dots,n-m-1\}$, every set $x \in D^t_{n,m,k}$ is obtained by adjoining any $k$ sets in $A_m \setminus A_{m-1}$ (hence in $B_{m,m-1}$), to a set $y \in A_{n-k}$ such that $y \subseteq A_{m-1}$ with cardinality at most $t-k$. As before, $y \notin A_{n-k-1}$ hence such sets $y$ are precisely those in $D^{t-k}_{n-k,m-1}$. This gives the next term in the recurrence relation.

Similarly, if $k = n-m$ every set $x \in D^t_{n,m,n-m}$ is obtained by adjoining any $n-m$ sets in $A_m \setminus A_{m-1}$ (hence in $B_{m,m-1}$), to any set $y \in A_m$ with cardinality at most $t-(n-m)$ (by Corollary~\ref{corollary-rka}). This gives the last term in the recurrence relation.
\end{proof}

\subsection{Hereditarily finite sets with atoms}
\label{sec:urelements}

Pure hereditarily finite sets can be generalized in a straightforward manner by allowing the presence of a  set $\mathcal{U}$ of \emph{urelements}, or \emph{atoms}, pairwise different and also different from sets (and in particular from $\emptyset$). The cumulative hierarchy of hereditarily finite sets with atoms $\mathcal{U}$ is defined as (see also \cite{DBLP:journals/ipl/PolicritiT11,Chiaruttini-Omodeo-08}):
\[V_0^{\mathcal{U}} := \mathcal{U}, \ \ V_{n+1}^{\mathcal{U}} := \mathcal{U} \cup \mathcal{P}(V_n^\mathcal{U}), \ \ V_\omega^{\mathcal{U}} := \bigcup_{n \in \omega} V_n^{\mathcal{U}}.\] 

The adjunctive hierarchy of hereditarily finite sets with urelements $\mathcal{U}$ can be analogously defined as:
\[A_0^{\mathcal{U}} := \{\emptyset\} \cup \mathcal{U},\]
\[A_{n+1}^\mathcal{U} := \{\emptyset\} \cup \mathcal{U} \cup \big\{ x \cup \{y\} \;|\; x \in A_n^\mathcal{U} \setminus \mathcal{U} ,y \in A_n^\mathcal{U} \big\},\]
so that $V_\omega^{\mathcal{U}} = \bigcup_{n \in \omega} A_n^{\mathcal{U}}$.

It is easy to see that, for a finite $\mathcal{U}$, all recurrences proposed in this paper generalize to $A_n^\mathcal{U}$, the only differences being in the initialization of the recurrences. For example, to obtain an analog of Theorem~\ref{thm-main}, denote by $B_{n,m}^{\mathcal{U}}$ the set $\left\lbrace x \in A_n^\mathcal{U} \setminus A_{n-1}^\mathcal{U} : ~ x \subseteq A_m^\mathcal{U} \right\rbrace$, and by $b_{n,m}^u$ the cardinality $|B_{n,m}^{\mathcal{U}}|$, where $u = |\mathcal{U}|$, so that 
\[|A_n^\mathcal{U}| = \sum_{k=0}^n b_{k,k-1}^u.\]

\begin{theorem}
For all natural numbers $n > m \geq 1$, the following recurrence relation holds
\begin{align*}
b_{n,m}^u &= b_{n,m-1}^u + \sum_{k=1}^{n-m-1} \left( b_{n-k,m-1}^u \binom{b_{m,m-1}^u}{k} \right) + \binom{b_{m,m-1}^u}{n-m} \left( 1+\sum_{k=1}^m b_{k,k-1}^u \right),
\end{align*}
where $b^u_{0,-1} = u+1$, and $b^u_{n,0} = \binom{u+1}{n}$ for all $n \geq 1$.
\end{theorem}

\begin{proof}
The proof follows closely the one of Theorem~\ref{thm-main}. The only difference is in the fact that the only set in $A_0$ allowed as first term in an adjunction is the empty set, thus giving $1+\sum_{k=1}^m b_{k,k-1}^u$ instead of $\sum_{k=0}^m b_{k,k-1}^u$.
\end{proof}

Numeric values of $|A_n^\mathcal{U}|$ are shown in Table~\ref{table-urelements}.

\begin{table*}[t]
\caption{Small values of $|A_n^\mathcal{U}|$ for small cardinalities of $\mathcal{U}$.\label{table-urelements}}
\small
\begin{center}
\begin{tabular}{|l|lllll|}
\hline
$n$ & $|\mathcal{U}| = 1$ & $|\mathcal{U}| = 2$ & $|\mathcal{U}| = 3$ & $|\mathcal{U}| = 4$ & $|\mathcal{U}| = 5$ \\
\hline
0 &	2 &	3 &	4 &	5 &	6 \\
1 &	4 &	6 &	8 &	10 &	12 \\
2 &	11 &	21 &	34 &	50 &	69 \\
3 &	86 &	328 &	898 &	2010 &	3932 \\
4 &	6707 &	102751 &	785834 &	3974665 &	15288832 \\
5 &	44661920 &	10540006012 &	617171670159 &	15793892739676 &	233717946472981 \\
\hline
\end{tabular}
\end{center}
\end{table*}

\section{Bounded adjunctive hierarchies}
\label{sec:bounded}

As Theorem~\ref{thm:asymp} shows, the growth of $a_n$, albeit slower than in the von Neumann's case, is still very fast, the $a_n$'s having an exponential number (to be more precise, $\Theta(2^n)$) of bits.

This observation motivates us to introduce bounded analogues of the adjunctive hierarchy in order to obtain slower asymptotic growth, ideally one in which the cardinalities of the different layers have a linear number of bits.

\begin{definition}
Given an unbounded sublinear function $f: \N \rightarrow \N$ (that is, for all $n$, $f(n) \leq n$, and there exists an $m$ such that $f(m)>n$), the adjunctive hierarchy bounded by $f$ is
\[
A^f_0 := \{\emptyset\}, \ \ A^f_{n+1} := \{\emptyset\} \cup \big\{ x \cup \{y\} \;|\; x \in A_n, y \in A_{f(n)} \big\}.
\]
We let $a^f_n = |A^f_n|$.
\end{definition}

Kirby's adjunctive hierarchy is obtained for $f$ the identity function. As in that case, it also holds that for any unbounded sublinear function $f$, $A^f_n \subseteq A^f_{n+1}$ for all $n \in \omega$, and $\bigcup_{n \in \omega}A^f_n = V_\omega$; that is, $A^f_n$ is a hierarchy of hereditarily finite sets. By changing the asymptotic growth of $f(n)$, it is possible to fine-tune the asymptotic behavior of $a^f_n$. In order to obtain numeric values for the $a^f_n$, it is possible to prove an analogue of the relation in Theorem~\ref{thm-main}.

\begin{definition} \label{def-inverse}
For every $n > m \geq 0$, let $B^f_{n,m}$ be the set $\left\lbrace x \in A^f_n \setminus A^f_{n-1} : ~ x \subseteq A^f_m \right\rbrace$, $b^f_{n,m} := |B^f_{n,m}|$ and $g: \N \rightarrow \N$ be such that $g(n) = \min\{t :~ f(t) \geq n \}$. We call this last function the \emph{inverse} of $f$.
\end{definition}

\begin{theorem}
\label{thm-bounded}
For all natural numbers $n > m \geq 0$, the following recurrence relation holds
\begin{align*}
b^f_{n,m} &= b^f_{n,m-1} + \sum_{k=1}^{n-g(m)-1} \left( b^f_{n-k,m-1} \binom{b^f_{m,m-1}}{k} \right) + \binom{b^f_{m,m-1}}{n-g(m)} \sum_{k=0}^{g(m)} b^f_{k,k-1},
\end{align*}
where $b^f_{0,-1} = 1$, and $b^f_{n,-1} = 0$, for all $n > 1$.
\end{theorem}

The proof follows step by step the one of Theorem~\ref{thm-main}, using the fact that adjoining $k$ elements not in $A^f_{m-1}$ to any set gives a set that is not in $A^f_{g(m)+k-1}$, and is left to the reader.

In Tables~\ref{table-2bound} and \ref{table-sqrtbound} in the Appendix we present numeric values of $a^f_n$ for some specific choices of $f$. When $f$ is not the identity function it can happen that $A^f_{n+1} = A^f_n$, in this case we feel free to skip the corresponding entries in the tables. These empirical results suggest that for $f(n) = \lfloor \log_2(n+1) \rfloor$ and possibly also for $f(n) = \lfloor \sqrt n \rfloor$, the number of digits needed to represent $a_n^f$ is less than $n$. 

As suggested from Table \ref{table-sqrtbound}, for functions $f$ growing sufficiently slowly and smoothly in some sense, the corresponding sequence $A^f_n$ is essentially the same (ignoring repetitions). In such cases we obtain a hierarchy equivalent to the following.

\begin{definition}
The \emph{minimally bounded} adjunctive hierarchy is
\begin{align*}
\bar{A}_0 &:= \{\emptyset\}, \\
\bar{A}_{n+1} &:= \{\emptyset\} \cup \big\{ x \cup \{y\} \;|\; x \in \bar{A}_n \wedge \exists m \left( \p(\bar{A}_m) \subseteq \bar{A}_n \wedge y \in \bar{A}_{m+1} \right)\!\!\big\}.
\end{align*}
We let $\bar{a}_n := |\bar{A}_n|$.
\end{definition}

In this sequence, the adjunction of a set $y \in \bar{A}_{m+1}$ is allowed only if all the sets obtainable \emph{\`a la} von Neumann from $\bar{A}_{m}$ are already in the hierarchy, that is, $\p(\bar{A}_{m}) \subseteq \bar{A}_n$, so that the adjunction of any set in $\bar{A}_{m}$ to sets in $\bar{A}_n$ would not produce any new set.

Observe that also $\bar{A}_{n-1} \subseteq \bar{A}_n$ holds for all $n \in \omega$.

\begin{lemma}
For all $n$, $\bar{A}_{\bar{a}_n} = \p(\bar{A}_n)$.
\end{lemma}

\begin{proof}
We prove the thesis by induction on $n$. For $n=0$, since $\bar{A}_0 = \{\emptyset\}$ then $\bar{a}_0 = 1$, and $\bar{A}_1 = \p(\bar{A}_0) = \left\{\emptyset, \{\emptyset\}\right\}$ is easily checked to be true.

Now suppose that $\bar{A}_{\bar{a}_{n-1}} = \p(\bar{A}_{n-1})$. We prove by another induction that for all $m$ such that $\bar{a}_{n-1} \leq m \leq \bar{a}_n$, the following holds:
\[
\bar{A}_m = \left\lbrace x \subseteq \bar{A}_n : ~ \left| x \setminus \bar{A}_{n-1} \right| \leq (m - \bar{a}_{n-1}) \right\rbrace .
\]
We will obtain the claim for $m = \bar{a}_{n}$, since plainly 
\[
\p(\bar{A}_n) = \left\lbrace x \subseteq \bar{A}_n : ~ \left| x \setminus \bar{A}_{n-1} \right| \leq (\bar{a}_n - \bar{a}_{n-1}) \right\rbrace.
\]
	For $m = \bar{a}_{n-1}$ it is true by the initial inductive hypothesis, since
\begin{align*}
\bar{A}_{\bar{a}_{n-1}} &= \p(\bar{A}_{n-1})\\
&= \left\lbrace x \subseteq \bar{A}_n : ~ x \subseteq \bar{A}_{n-1} \right\rbrace \\
&= \left\lbrace x \subseteq \bar{A}_n : ~ \left| x \setminus \bar{A}_{n-1} \right| \leq 0 \right\rbrace.
\end{align*}

Let $m > \bar{a}_{n-1}$ and assume that the hypothesis holds for $\bar{A}_{m-1}$. Since $m \leq \bar{a}_n$, this implies that $\bar{A}_{m-1} \subsetneq \p(\bar{A}_n)$. Therefore, by the definition and the fact that $\bar{A}_{\bar{a}_{n-1}} = \p(\bar{A}_{n-1})$ holds, we have that
\[
\bar{A}_{m} := \{\emptyset\} \cup \big\{ x \cup \{y\} \;|\; x \in \bar{A}_{m-1} ~ \wedge ~y \in \bar{A}_{n-1} \big\}.
\]
This proves the claim, since the inductive hypothesis holds for $\bar{A}_{m-1}$.
\end{proof}

The last result shows that for infinitely many $n$ (i.e. if there exists $m$ such that $\bar{a}_m = n$), $\bar{a}_n = 2^n$ hence the asymptotic growth of $\bar{a}_n$ is $\Theta(2^n)$. It also follows that $\bar{A}_{|V_n|} = V_{n+1}$, thus the sequence $\bar{a}_n$ contains the usual von Neumann's hierarchy as a subsequence.

\begin{corollary} \label{corollary-bounded}
The sequence $\bar{A}_n$ is obtained as $A^{\bar{f}}_n$ where the function $\bar{f}$ is defined as
\[
\bar{f}(n) = \min \left\lbrace m: ~ \bar{a}_m > n \right\rbrace .
\]
Thus, the corresponding inverse $\bar{g}$ is such that $\bar{g}(n) = \bar{a}_{n-1}$, $\bar{g}(0) = 0$.
\end{corollary}

As previously mentioned, any other surjective function $f$ with inverse $g$ (as defined in Definition \ref{def-inverse}) such that $g(n+1) - g(n) \geq \bar{g}(n+1) - \bar{g}(n)$ holds for all $n$ will produce the same sequence, possibly with repetitions. Moreover, it is possible to combine Theorem \ref{thm-bounded} and Corollary \ref{corollary-bounded} to produce a recursive relation for $\bar{b}_{n,m} = b^{\bar{f}}_{n,m}$:
\begin{align*}
\bar{b}_{n,m} &= \bar{b}_{n,m-1} + \sum_{k=1}^{n-\bar{a}_{m-1}-1} \left( \bar{b}_{n-k,m-1} \binom{\bar{b}_{m,m-1}}{k} \right) + \bar{a}_{\bar{a}_{m-1}} \binom{\bar{b}_{m,m-1}}{n-\bar{a}_{m-1}},
\end{align*}
where $\bar{a}_m$ is definable from the values $\bar{b}_{n,m}$ as $\sum_{k=0}^m \bar{b}_{m,m-1}$. This recursive relation can be implemented in a dynamic programming algorithm analogous to the one for $b_{n,m}$; small numerical values of it are shown in Table~\ref{table-minbound} in the Appendix.

\section{Conclusions and future work}
\label{sec:conclusions}
In this paper we solved a counting problem proposed by Kirby in 2008, and showed that our method can be extended by imposing restrictions on the rank or on the cardinality of the hereditarily finite sets making up a certain level of the adjunctive hierarchy. We also showed that it can be easily extended to hereditarily finite sets with atoms.

An asymptotic analysis of the recurrences revealed that the adjunctive hierarchy, even though much slower than the usual cumulative hierarchy, is still growing too fast to allow for an efficient encoding of hereditarily finite sets by numbers. 

For this reason, we proposed the minimally bounded adjunctive hierarchy, in which an $(n+1)$th level is obtained by the adjunction of a set $y \in \bar{A}_{m+1}$ to a set $x \in \bar{A}_n$ if all the sets in $\p(\bar{A}_m)$ are already present in $\bar{A}_n$. This is a natural combination of the adjunctive hierarchy and the cumulative hierarchy; more importantly, we proved that its asymptotic growth is $\Theta(2^n)$. To the best of our knowledge, this is the first result of this kind. 

The next step to take is to devise efficient ranking/unranking algorithms for our minimally bounded adjunctive hierarchy, as Ackermann's encoding \cite{Ack37} is for the cumulative hierarchy, and as the algorithms in \cite{DBLP:journals/ipl/RizziT13} are for the hierarchy of transitive sets with a given number of elements.

A different open problem can be formulated as follows. Note that all transitive sets with $n$ elements belong to $V_n$, and recall that the number of transitive sets with $n$ elements has already been obtained in \cite{P62}; it would thus be interesting to find an analog of this result for the adjunctive hierarchy, namely, to find the number of \emph{transitive} sets in $A_n$ with $n$ elements.


\section*{Acknowledgement}
We thank Laurence Kirby for telling us about this interesting problem, and Eugenio Omodeo for helpful discussions. The second author is partially supported by the Academy of Finland under grant 250345 (CoECGR), and by the European Science Foundation, activity ``Games for Design and Verification''. The third author is supported financially by the National Research Foundation of South Africa, grant number 70560.



\bibliographystyle{IEEEtran}
\bibliography{adjunction}   
%
%
%

\newpage
\appendix
\section*{A bounded growth of the adjunctive hierarchy}

\begin{table}[h]
\caption{First values of $a^f_n$ for $f(n) = \lceil n/2 \rceil$.}
\begin{center}
\setlength{\tabcolsep}{.16cm}
\begin{tabular}{|l|l|}
\hline
$a^f_{ 0 }$ &	1 \\
$a^f_{ 1 }$ &	2 \\
$a^f_{ 2 }$ &	4 \\
$a^f_{ 4 }$ &	12 \\
$a^f_{ 5 }$ &	16 \\
$a^f_{ 8 }$ &	144 \\
$a^f_{ 9 }$ &	592 \\
$a^f_{ 10 }$ &	3856 \\
$a^f_{ 11 }$ &	12112 \\
$a^f_{ 12 }$ &	25232 \\
$a^f_{ 13 }$ &	40160 \\
$a^f_{ 14 }$ &	52832 \\
$a^f_{ 15 }$ &	60752 \\
$a^f_{ 16 }$ &	7840528 \\
$a^f_{ 17 }$ &	502084400 \\
$a^f_{ 18 }$ &	246203916272 \\
$a^f_{ 19 }$ &	60472296567808 \\
$a^f_{ 20 }$ &	207302387302931456 \\
$a^f_{ 21 }$ &	355632667741263729920 \\
$a^f_{ 22 }$ &	3343198667129228884545792 \\
$a^f_{ 23 }$ &	15829569100117020469497511168 \\
$a^f_{ 24 }$ &	258028007928627813480157366817024 \\
$a^f_{ 25 }$ &	2143825383084631588989060293305465472 \\
$a^f_{ 26 }$ &	44114691903811742239796481826048657798272 \\
$a^f_{ 27 }$ &	472009200002288950265751813320485308731259904 \\
$a^f_{ 28 }$ &	9485240116915376700878425362559719242896317641728 \\
$a^f_{ 29 }$ &	102586446112048504015292656228608097259346546351742208 \\
\hline
\end{tabular}
\end{center}
\label{table-2bound}
\end{table}

\begin{table}[h!]
\caption{First values of $a^{f_1}_n$ and $a^{f_2}_n$ for $f_1 = \lfloor \sqrt n \rfloor$ and $f_2 = \lfloor \log_2(n+1) \rfloor$.}
\begin{center}
\begin{tabular}{|l|c|l|}
\hline
$a^{f_1}_{ 0 }$ &	1 & $a^{f_2}_{ 0 }$ \\
$a^{f_1}_{ 1 }$ &	2 & $a^{f_2}_{ 1 }$\\
$a^{f_1}_{ 2 }$ &	4 & $a^{f_2}_{ 2 }$\\
$a^{f_1}_{ 5 }$ &	12 & $a^{f_2}_{ 4 }$\\
$a^{f_1}_{ 6 }$ &	16 & $a^{f_2}_{ 5 }$\\
$a^{f_1}_{ 26 }$ &	144 & $a^{f_2}_{ 16 }$\\
$a^{f_1}_{ 27 }$ &	592 & $a^{f_2}_{ 17 }$\\
$a^{f_1}_{ 28 }$ &	1488 & $a^{f_2}_{ 18 }$\\
$a^{f_1}_{ 29 }$ &	2608 & $a^{f_2}_{ 19 }$\\
$a^{f_1}_{ 30 }$ &	3504 & $a^{f_2}_{ 20 }$\\
$a^{f_1}_{ 31 }$ &	3952 & $a^{f_2}_{ 21 }$\\
$a^{f_1}_{ 32 }$ &	4080 & $a^{f_2}_{ 22 }$\\
$a^{f_1}_{ 33 }$ &	4096 & $a^{f_2}_{ 23 }$\\
$a^{f_1}_{ 37 }$ &	20480 & $a^{f_2}_{ 32 }$\\
$a^{f_1}_{ 38 }$ &	45056 & $a^{f_2}_{ 33 }$\\
$a^{f_1}_{ 39 }$ &	61440 & $a^{f_2}_{ 34 }$\\
$a^{f_1}_{ 40 }$ &	65536 & $a^{f_2}_{ 35 }$\\
$a^{f_1}_{ 677 }$ &	8454144 & $a^{f_2}_{ 65536 }$\\
$a^{f_1}_{ 678 }$ &	541130752 & $a^{f_2}_{ 65537 }$\\
$a^{f_1}_{ 679 }$ &	22913548288 & $a^{f_2}_{ 65538 }$\\
$a^{f_1}_{ 680 }$ &	722051596288 & $a^{f_2}_{ 65539 }$\\
$a^{f_1}_{ 681 }$ &	18060675186688 & $a^{f_2}_{ 65540 }$\\
$a^{f_1}_{ 682 }$ &	373502458789888 & $a^{f_2}_{ 65541 }$\\
$a^{f_1}_{ 683 }$ &	6568344973017088 & $a^{f_2}_{ 65542 }$\\
$a^{f_1}_{ 684 }$ &	100265338000703488 & $a^{f_2}_{ 65543 }$\\
$a^{f_1}_{ 685 }$ &	1349558578369855488 & $a^{f_2}_{ 65544 }$\\
$a^{f_1}_{ 686 }$ &	16216148138762764288 & $a^{f_2}_{ 65545 }$\\
$a^{f_1}_{ 687 }$ &	175694108877523058688 & $a^{f_2}_{ 65546 }$\\
$a^{f_1}_{ 688 }$ &	1730604226080435929088 & $a^{f_2}_{ 65547 }$\\
$a^{f_1}_{ 689 }$ &	15605186810352581541888 & $a^{f_2}_{ 65548 }$\\
$a^{f_1}_{ 690 }$ &	129574972324016634789888 & $a^{f_2}_{ 65549 }$\\
$a^{f_1}_{ 691 }$ &	995745342227863439474688 & $a^{f_2}_{ 65550 }$\\
$a^{f_1}_{ 692 }$ &	7113073579673781497561088 & $a^{f_2}_{ 65551 }$\\
\hline
\end{tabular}
\end{center}
\label{table-sqrtbound}
\end{table}

\newpage
\phantom{.}

\begin{table}[t!]
\caption{First values of $\bar{a}_n$.}
\begin{center}
\begin{tabular}{|l|l|}
\hline
$\bar{a}_{ 0 }$ &	1 \\
$\bar{a}_{ 1 }$ &	2 \\
$\bar{a}_{ 2 }$ &	4 \\
$\bar{a}_{ 3 }$ &	12 \\
$\bar{a}_{ 4 }$ &	16 \\
$\bar{a}_{ 5 }$ &	144 \\
$\bar{a}_{ 6 }$ &	592 \\
$\bar{a}_{ 7 }$ &	1488 \\
$\bar{a}_{ 8 }$ &	2608 \\
$\bar{a}_{ 9 }$ &	3504 \\
$\bar{a}_{ 10 }$ &	3952 \\
$\bar{a}_{ 11 }$ &	4080 \\
$\bar{a}_{ 12 }$ &	4096 \\
$\bar{a}_{ 13 }$ &	20480 \\
$\bar{a}_{ 14 }$ &	45056 \\
$\bar{a}_{ 15 }$ &	61440 \\
$\bar{a}_{ 16 }$ &	65536 \\
$\bar{a}_{ 17 }$ &	8454144 \\
$\bar{a}_{ 18 }$ &	541130752 \\
$\bar{a}_{ 19 }$ &	22913548288 \\
$\bar{a}_{ 20 }$ &	722051596288 \\
$\bar{a}_{ 21 }$ &	18060675186688 \\
$\bar{a}_{ 22 }$ &	373502458789888 \\
$\bar{a}_{ 23 }$ &	6568344973017088 \\
$\bar{a}_{ 24 }$ &	100265338000703488 \\
$\bar{a}_{ 25 }$ &	1349558578369855488 \\
$\bar{a}_{ 26 }$ &	16216148138762764288 \\
$\bar{a}_{ 27 }$ &	175694108877523058688 \\
$\bar{a}_{ 28 }$ &	1730604226080435929088 \\
$\bar{a}_{ 29 }$ &	15605186810352581541888 \\
$\bar{a}_{ 30 }$ &	129574972324016634789888 \\
$\bar{a}_{ 31 }$ &	995745342227863439474688 \\
$\bar{a}_{ 32 }$ &	7113073579673781497561088 \\
$\bar{a}_{ 33 }$ &	47415471379317476939071488 \\
$\bar{a}_{ 34 }$ &	295946924477120265495052288 \\
$\bar{a}_{ 35 }$ &	1734813231885452199240204288 \\
$\bar{a}_{ 36 }$ &	9576634607260861238151282688 \\
$\bar{a}_{ 37 }$ &	49906001680620107723979685888 \\
$\bar{a}_{ 38 }$ &	246053377901049170177781465088 \\
$\bar{a}_{ 39 }$ &	1150036937873461371051824447488 \\
$\bar{a}_{ 40 }$ &	5104965012752764749875762495488 \\
$\bar{a}_{ 41 }$ &	21557465804250666805783344775168 \\
$\bar{a}_{ 42 }$ &	86734680478261586488801843806208 \\
$\bar{a}_{ 43 }$ &	332959713691191727513538395701248 \\
$\bar{a}_{ 44 }$ &	1221128583494975450495623815036928 \\
$\bar{a}_{ 45 }$ &	4283779858680436564226952847228928 \\
\hline
\end{tabular}
\end{center}
\label{table-minbound}
\end{table}

\end{document}